\RequirePackage{fix-cm}
\documentclass[sn-mathphys,Numbered,runningheads, UKenglish, autoref, thm-restate]{sn-jnl}

\usepackage[utf8]{inputenc}
 
\usepackage{amsmath,amsfonts,amssymb}
\usepackage{hyperref}
\usepackage{cleveref}
\usepackage{lmodern}
\usepackage{framed}
\usepackage[dvipsnames]{xcolor}

\usepackage{xfrac}
\usepackage{mathrsfs}
\DeclareFontFamily{U}{rsfs}{\skewchar\font127 }
\DeclareFontShape{U}{rsfs}{m}{n}{%
   <-6> rsfs5
   <6-8> rsfs7
   <8-> rsfs10
}{}
\usepackage{multirow}%

\usepackage{amsthm}%
\usepackage{mathrsfs}%
\usepackage[title]{appendix}%

\usepackage{textcomp}%
\usepackage{manyfoot}%
\usepackage{booktabs}%
\usepackage{algorithm}%
\usepackage{algorithmicx}%
\usepackage{algpseudocode}%
\usepackage{listings}%

\usepackage[disableredefinitions]{complexity}
\usepackage{xspace}

\usepackage{enumitem}
\usepackage{bold-extra}

\usepackage{newunicodechar}
\newunicodechar{≤}{\ensuremath{\leq}}
\newunicodechar{⊕}{\ensuremath{\oplus}}

\usepackage{tikz}

\newtheorem{theorem}{Theorem}
\theoremstyle{definition}%
 \newtheorem{definition}{Definition}%

 \newtheorem{claim}{Claim}

\usepackage[appendix=inline]{apxproof}
\newtheoremrep{theorem}{Theorem}
\newtheoremrep{lemma}{Lemma}
\newtheoremrep{corollary}{Corollary}
\newtheoremrep{observation}{Observation}
\Crefname{observation}{Observation}{Observations}
\usepackage[T1]{fontenc}
\usepackage{graphicx}
\usepackage{lineno}
\begin{document}


%
\title{Temporal Orienteering with Changing Fuel Costs}

%

\author[1]{Timothée Corsini}\email{timothee.corsini@labri.fr}

\author*[2]{Jessica Enright}\email{jessica.enright@glasgow.ac.uk}

\author*[2]{Laura Larios-Jones}\email{l.larios-jones.1@research.gla.ac.uk}
\author*[2]{Kitty Meeks}\email{kitty.meeks@glasgow.ac.uk}
\affil[1]{\orgdiv{LaBRI, CNRS}, \orgname{Univ. Bordeaux}} 

\affil[2]{\orgdiv{School of Computing Science}, \orgname{University of Glasgow}}
%

%
%

%
%
%

\abstract{The Orienteering problem asks whether there exists a walk which visits a number of sites without exceeding some fuel budget. In the variant of the problem we consider, the cost of each edge in the walk is dependent on the time we depart one endpoint and the time we arrive at the other endpoint. This mirrors applications such as travel between orbiting objects where fuel costs are dependent on both the departure time and the length of time spent travelling.
In defining this problem, we introduce a natural generalisation of the standard notion of temporal graphs: the pair consisting of the graph of the sites and a cost function, in which costs as well as shortest travel times between pairs of objects change over time. We believe this model is likely to be of independent interest.
The problem of deciding whether a stated goal is feasible is easily seen to be NP-complete; we investigate three different ways to restrict the input which lead to efficient algorithms. These include the number of times an edge can be used, an analogue of vertex-interval-membership width, and the number of sites to be visited.}

\maketitle

\section{Introduction}
The \textsc{Orienteering} problem, also known as \textsc{Selective Travelling Salesman}~\cite{laporte_selective_1990} or the \textsc{Bank Robber}~\cite{awerbuch1995improved} problem, asks for a path in a graph which visits some number of vertices such that the sum of the costs on the edges of the path is at most a given budget. Note that, when the number of vertices to be visited is the same as the number of vertices in the graph, \textsc{Orienteering} becomes the classical \textsc{Travelling Salesman Problem} (TSP)~\cite{gavish_travelling_1978}. As with TSP, many versions of \textsc{Orienteering} have been studied (see~\cite{vansteenwegen_orienteering_2011} for a survey).

Our variation of the \textsc{Orienteering} problem is \textsc{Changing Cost Temporal Orienteering}. In this variant, we look for a walk which visits the designated number of vertices with fuel costs at most the given budget. The cost of a walk is given by the sum of the costs of the edges at the time they are traversed, where the cost of an edge is given by a function dependent on the edge, departure time at one endpoint, and arrival time at the other endpoint.
Note that looking for a walk rather than a path allows us to revisit vertices. Our variant of the problem is inspired by scheduling travel to objects in near-earth orbit~\cite{lopez-ibanez_asteroid_2022}. Since these objects are in orbit and the distances between them change over time, so too do the associated fuel costs. In addition, one may choose to expend more fuel to arrive at the destination faster, or take a cheaper and slower journey. Our model allows both waiting at a vertex until the journey to the next vertex is cheaper, and for making longer, but more fuel-efficient, journeys.

In order to describe this problem formally, we introduce a new notion of \emph{temporal cost graphs}, a generalisation of the well-studied notion of temporal graphs which allows for the cost of travelling between two vertices to change over time; we believe that this mathematical object will prove useful for the study of other problems.  We observe that our problem -- \textsc{Changing Cost Temporal Orienteering} -- generalises star exploration, known to be NP-complete even under strong restrictions \cite{AMSR21,bumpus2021edge}.  We therefore focus our attention on restricted cases in which we can hope to recover tractability.  Informally, we design efficient algorithms by restricting any one of the following: 
\begin{itemize}
    \item the number of times at which it is possible to either enter or leave a vertex with finite travel cost on trees (Section \ref{sec:trees});
    \item the maximum, taken over all timesteps, of the number of vertices possible for the orienteer's current location, assuming they have only used finite fuel and are not permanently stranded on the current object (Section \ref{sec:vitw});
    \item the target number of objects we must visit (Section \ref{sec:k-small}).
\end{itemize}
We note that our second restriction (found in Section~\ref{sec:vitw}) is a generalisation of vertex-interval-membership width. 

\subsection{Related Work}
The \textsc{Orienteering} problem \cite{golden_orienteering_1987} is described by Vansteenwegan \cite{vansteenwegen_orienteering_2011} as a combination of the \textsc{Knapsack} Problem \cite{kellerer_knapsack_2004} and TSP. When the vertices are weighted, the problem has also been referred to as \textsc{Prize Collecting TSP}~\cite{balas1989prize} or \textsc{TSP with profits}~\cite{feillet_traveling_2005}. Note that we allow the walk to revisit vertices, however some variants of the problem do not allow this and instead require a tour of the visited vertices. Blum et al.~\cite{blum2007approximation} give a constant factor approximation algorithm for the \textsc{Orienteering} problem that demands a tour given a specified starting vertex. Similar methods are used by Chekuri et al.~\cite{chekuri2012improved} for the variant where vertices can be revisited to obtain a $(2+\varepsilon)$-approximation.

Fomin and Lingas~\cite{FL02} introduce the \textsc{Time-Dependent Orienteering} problem which, given a clique and a travel time function which depends on the current absolute time, asks for a tour that visits $k$ vertices before a deadline. The travel time function used in their work only takes in a starting time to calculate the cost, rather than both a starting and finishing time as we do here. In a survey of different notions of time-varying graphs, Casteigts et al.~\cite{casteigts2012time} gives a model which has a latency function over the edge-time pairs which describes how long it takes to cross an edge at a given time. This latency function operates in the same way as the travel time function used by Fomin and Lingas. They show that their generalisation of the \textsc{Orienteering} problem is $(2+\epsilon)-$approximable if the ratio between the maximum and minimum travel times is bounded by a constant.  Our problem is related: we also compute on a graph and with a travel time function, but our version allows the fuel cost to vary arbitrarily with the travel time. Our model also explicitly uses a different underlying graphs, and allows vertices to be revisited. As a result, their results cannot be easily adapted to our setting.

Other time-dependent variations of \textsc{Orienteering} are considered by Buchin et al.~\cite{buchin2025orienteering}. In their work, they search for a walk on a graph whose vertices are labelled with intervals at which each vertex can be visited. They show that the problem is weakly NP-hard when the underlying graph is a path and with unit costs on the edges. Buchin et al.~also consider a variant of the problem where edges, rather than vertices, are labelled with intervals in which they can be used. This is an analogue of \textsc{Orienteering} for temporal graphs. For this problem, they show tractability when the underlying graph is a path, and hardness when it is a tree. They extend some of their structural results to the original \textsc{Orienteering} problem.

Capturing our space-related motivation, López-Ibáñez et al.~\cite{lopez-ibanez_asteroid_2022} present the \textsc{Asteroid Routing} problem which asks for the optimal sequence for visiting a set of asteroids, starting from Earth’s orbit, in order to minimize both the cost and the time taken. They give four algorithms and compare their efficacy in an experimental study.

\subsection{Problem definition and basic observations}

A \emph{temporal graph} is usually defined as a pair $(G = (V,E), \lambda)$, where $G$ is a static graph and the function $\lambda : E \to 2^\tau$ attributes to each edge a set of time labels at which the edge is said to be \emph{active}.  To define our problem, we begin by introducing a generalisation of this notion in which each edge of the static graph has an associated traversal cost, which depends on both the departure and arrival times as well as the direction of travel (we include the arrival time as well as departure time as an input to the cost function as one can imagine a scenario in which the orienteer has a choice between a fast but costly route or a slower but more fuel-efficient trajectory departing at the same time). This also generalises the time-varying graph model given in the survey by Casteigts et al.~\cite{casteigts2012time}. This leads us to the following definition.

\begin{definition}[Temporal cost graph]
    A \emph{temporal cost graph} is a static graph $G$ paired with a cost function $F:V\times V\times [T]\times [T]\to \mathbb{N}_0 \cup \{\infty\}$ which takes an ordered tuple consisting of a starting vertex $v_1$, a destination vertex $v_2$, a departure time $t_1$ and an arrival time $t_2$ and assigns a fuel cost.  We require that $F(v_1,v_2,t_1,t_2) = \infty$ whenever $t_1 \geq t_2$, $F(v,v,t_1,t_2) = 0$ for all $v$ and any $t_1 < t_2$, and that $F(v_1,v_2,t_1,t_2) > 0$ whenever $v_1 \neq v_2$.
\end{definition}
\noindent
We assume throughout that temporal cost graphs are encoded in such a way (e.g.~using a hash table) that, given any tuple $(u,v,t_1,t_2)$, we can look up the value of $F(u,v,t_1,t_2)$ in constant time.

Notice that our definition of temporal cost graphs generalises the standard definition of temporal graphs.  Given a temporal graph $(G,\lambda)$, we can obtain an equivalent temporal cost graph by taking the same underlying graph $G$ and setting $F(u,v,t_1,t_2) = 1$ if and only if the edge $uv$ is active at time $t_1$, and $t_2 = t_1 + 1$. If this is not the case, the quadruple is assigned infinite cost by $F$.

We now define the key concept of a valid walk in a temporal cost graph. In a temporal cost graph generated from a standard temporal graph, this is equivalent to the notion of a strict temporal walk.

\begin{definition}[Valid walk]
    A \emph{valid walk} on a temporal cost graph $(G,F)$ is a sequence of quadruples $(u_1,v_1,t_1,t_1'),(u_2,v_2,t_2,t_2'),\dots,(u_{\ell},v_{\ell},t_{\ell},t_{\ell}')$ such that:
    \begin{itemize}
        \item for each $1 \le i < \ell$, $v_i = u_{i+1}$,
        \item for each $1 \le i \le \ell$, $t_i < t_i'$ and $F(u_i,v_i,t_i,t_i') < \infty$, and
        \item for each $1 \le i < \ell$, $t_i' \le t_{i+1}$.
    \end{itemize}
    We define the \emph{cost} of a valid walk $W$ to be $\sum_{(u_i,v_i,t_i,t_i') \in W} F(u_i,v_i,t_i,t_i')$.
\end{definition}

Note that, since each quadruple $(u_i,v_i,t_i,t'_i)$ in a valid walk is required to have finite cost, it must be the case that $t_i<t_i'$. Also note that we use the convention that staying at a vertex incurs no cost.

We can now give the formal definition of our problem.
\begin{framed}
\noindent
\textbf{\textsc{Changing Cost Temporal Orienteering (CCTO)}}\\
\textit{Input:} A temporal cost graph $(G,F)$, source vertex $v_s$, sink vertex $v_t$, a target $k \in \mathbb{N}$ and a budget $f \in \mathbb{N}$.\\
\textit{Question:} Is there a valid walk $W$ from $v_s$ to $v_t$ that visits at least $k$ distinct vertices and has cost at most $f$?
\end{framed}
For avoidance of doubt: we count the source and sink vertices in the number of vertices visited by a valid walk. We also assume without loss of generality that the underlying graph of the input temporal cost graph is connected. As is standard in the \textsc{Orienteering} problem, the start and end vertex of the walk is specified. We note that all of the algorithmic results in this work are generalisable to the unrooted case where the start and end are not prescribed by adding a quadratic factor which tries all pairs of start and end.

We now define two important properties of temporal cost graphs which we shall refer to when designing our algorithms.  These are the analogues of the notions of lifetime and temporality, respectively, for temporal graphs.

\begin{definition}[Lifetime]
    The \emph{lifetime} $T$ of a temporal cost graph $(G,F)$ is the latest arrival time $t'$ of an edge such that there exist vertices $v$, $u$ and a departure time $t$ with $F(v,u,t,t')<\infty$.
\end{definition}

\begin{definition}[Maximum traversal number]
    Let $e = uv$ be an edge in $G$.  The \emph{maximum traversal number} of $e$ in the temporal cost graph $(G,F)$ is the largest integer $j$ such that there exists a sequence $(t_1,t_1'),(t_2,t_2'),\dots,(t_j,t_j')$ of pairs of times such that:
    \begin{itemize}
        \item for each $1 \le i \le j$, we have $\min \{F(u,v,t_i,t_i'),F(v,u,t_i,t_i')\} < \infty$, and
        \item for each $1 \le i \le j-1$, we have $t_i' \le t_{i+1}$.
    \end{itemize}
\end{definition}

\noindent
Crucially, the maximum traversal number of an edge $e$ in $(G,F)$ provides an upper bound on the number of times that $e$ can be used in any valid walk.

We use these parameters to show (in)tractability.

\begin{definition}[Fixed-parameter tractable, para-NP-hard]
    Given a problem $P$, an instance $I$ and a parameter $k$, an algorithm solving $P$ is fixed-parameter tractable (fpt) with respect to $k$ if it runs in time $O(f(k)\text{poly}(|I|))$ for some function $f$. The complexity class of problem-parameter pairs for which there exists an fpt-algorithm is FPT.\footnote{We use capitalisation to distinguish between the complexity class (FPT) and the descriptor (fpt).}
    A problem $P$ is para-NP-hard with respect to a parameter $k$ if it is NP-hard when $k$ is a constant.
\end{definition}

We conclude this section by observing that CCTO is NP-complete.  Recall that our definition of temporal cost graphs generalises the standard definition of temporal graphs so, if a problem is hard on temporal graphs, the analogous problem is hard on temporal cost graphs.  It is therefore easy to deduce NP-hardness of CCTO from the NP-hardness of the following problem on temporal graphs, which is known to be NP-hard even when every edge is active at no more than four timesteps \cite{AMSR21,bumpus2021edge}; here $S_n$ is a star with $n$ leaves, and our goal is to visit every leaf before returning to the centre.

\begin{framed}
\noindent
\textbf{\textsc{StarExp}}\\
\textbf{Input:} A temporal star $(S_n, \lambda)$ with centre $x$.\\
\textbf{Question:} Is there an ordering $v_1,\dots,v_n$ of the leaves of $S_n$ and a sequence $t_1,t_2,\dots,t_{2n}$ of times with $t_i < t_{i+1}$ for each $1 \le i \le 2n-1$ such that, for each $1 \le j \le n$, we have $t_{2j-1},t_{2j} \in \lambda(xv_j)$?
\end{framed}

Given the input $(S_n,\lambda)$ to an instance of \textsc{StarExp}, we can easily construct an equivalent instance of \textsc{CCTO}.  Suppose that the latest time at which any edge in $(S_n,\lambda)$ is active is $\tau$.  We define a cost function $F: V(S_n) \times V(S_n) \times [\tau+1] \times [\tau+1]$ by setting $F(u,v,t_1,t_2) = 1$ if $t_1 \in \lambda(uv)$ and $t_2 = t_1 + 1$, and $F(u,v,t_1,t_2) = \infty$ otherwise.   It is then immediate that $(S_n,x,x,F,n+1,2n)$ is a yes-instance for \textsc{CCTO} if and only if $(S_n,\lambda)$ is a yes-instance for \textsc{StarExp} -- recall that $S_n$ has $n+1$ vertices. We further note that \textsc{CCTO} is clearly in NP (with the sequence of edges traversed together with the departure and arrival times acting as a certificate) so we conclude that \textsc{CCTO} is NP-complete. This gives us the following result.

\begin{theorem}\label{thm:hardness}
    \textsc{CCTO} is NP-complete and para-NP-hard with respect to maximum traversal number, even when the underlying graph is a star.
\end{theorem}

It is natural to ask -- especially as we will show in Section \ref{sec:k-small} that \textsc{CCTO} is in FPT parameterised by the number $k$ of vertices to be visited -- whether NP-hardness relies on the fact that our target number of objects to visit is equal to the number of vertices.  In fact, we can easily adapt the reduction described above to give NP-hardness when $k = \mathcal{O}(n^c)$ for any $c \in (0,1]$.  Let $(S_{n'},x,x,F,n'+1,2n')$ be an instance of \textsc{CCTO} constructed via the reduction above.  We then add $\Theta((n')^{1/c})$ additional leaves so that the total number of vertices is $n$ and we have $n' = \mathcal{O}(n^c)$.  We extend $F$ to these additional vertices by making it possible to traverse the edges to the new leaves with cost $1$ only at times $\tau+1$ and $\tau+2$.  The key observation is that we can visit at most one of these new vertices in any walk of finite cost, so if we increase our budget by $2$ and the target number of vertices to visit by one, this instance is again equivalent to the original instance of \textsc{StarExp}.



\section{Trees and bounded maximum traversal number}\label{sec:trees}

We begin by designing algorithms to solve \textsc{CCTO} on temporal cost graphs with restricted underlying graphs and cost functions. This uses a generalisation of the notion of the time expanded graph of a temporal graph, and the fact that this is a directed acyclic graph.  We consider several different restrictions on the input, each of which allows us to ensure that we have a path consisting of $h$ edges if and only if we visit at least $g(h)$ distinct vertices in the corresponding walk on the original graph for some computable function $g$.

We begin by defining the time expanded cost graph; a generalisation of the \emph{time expanded graph} of a temporal graph. Here, for a temporal graph $(G,\lambda)$, the time expanded graph is a directed graph with vertices $u_t$ for each $u \in V(G)$ and $t \in [\tau]$. For all $u \in V(G)$, there is an arc from $u_t$ to $u_{t'}$ if $t' > t$ and for all $v \in V(G)$, there is an arc from $u_t$ to $v_{t+x}$ if $t \in \lambda(u,v)$ and $x$ is the travel time of such edge at time $t$. Assuming that the travel times are strictly greater than~$0$, the time expanded graph is a directed acyclic graph (DAG).

  \begin{definition}[Time expanded cost graph]\label{def:time-expanded-cost}
     The \emph{time expanded cost graph} of a temporal cost graph $(G,F)$ with lifetime $T$ is a directed graph consisting of a vertex $u_t$ for each $u \in V(G)$ and $t \in [T]$. For all $u,v \in V(G)$, there is an arc from $u_{t_1}$ to $v_{t_2}$ if $F(u,v,t_1,t_2)<\infty$. We call such a graph where the arcs $v_{t_1}u_{t_2}$ are labelled with $F(v,u,t_1,t_2)$ \emph{weighted}.
 \end{definition}

Note that, since the cost of staying at a vertex is $0$, there is an arc $u_t$ to $u_{t'}$ in the time expanded cost graph if $t' > t$. As with a time expanded graph, a time expanded cost graph is a DAG. We construct the time-expanded cost graph of a temporal cost graph by adding the necessary vertices and adding arcs where needed for every pair of vertices and pair of increasing times.

 \begin{lemmarep}\label{lem:expanded-graph-poly}
     Given a temporal cost graph $(G,F)$ with $n$ vertices and lifetime $T$, we can construct its time expanded cost graph in $O(n^2T^2)$ time.
 \end{lemmarep}
 \begin{proof}
     We begin by creating a vertex for each vertex time pair. There are exactly $nT$ vertices to add. To find all arcs in the time expanded cost graph, we determine the cost of every edge at every pair of increasing times. This is bounded by $n^2T^2$. Assuming we can retrieve $F(v_1,v_2,t_1,t_2)$ for inputs $v_1,v_2,t_1,t_2$ in constant time, we can therefore construct the graph in $O(n^2T^2)$.
 \end{proof}

Our first result is in the setting where the underlying graph $G$ is a tree, the source and sink are the same vertex $s$, and the maximum traversal number of any edge in the temporal cost graph $(G,F)$ is at most 3. The intuition of the following result is that, under these conditions, any edge can only be used at most twice on a valid walk. We use this fact paired with an auxiliary graph on which we look for a path from a source vertex to a sink vertex. The auxiliary graph is similar to the time expanded cost graph with additional features that we use to keep track of the number of vertices visited by the walk. We note here that if the maximum traversal number is increased to~4, we have hardness on trees by Theorem~\ref{thm:hardness}. Therefore, this result gives us algorithmic tightness in some sense.

To construct the auxiliary graph $H$, for each $u\in V(G)$, $t\in [T]$ and $i\in [k]$ we add a vertex $u_{t,i}$. If a vertex $u$ is closer to $s$ than a vertex $v$ in $G$, we add an arc labelled $F(u,v,t_1,t_2)$ from $u_{t_1,i}$ to $v_{t_2,i'}$ where $i'=\min\{i+1,k\}$ if and only if $F(u,v,t_1,t_2)$ is finite. If $u$ is further from $s$ than $v$ is, we add an arc labelled $F(u,v,t_1,t_2)$ from $u_{t_1,i}$ to $v_{t_2,i}$
    if and only if $F(u,v,t_1,t_2)$ is finite. That is, we only increment $i$ if we are visiting a new vertex and we have not yet visited $k$ vertices.
A sketch of the auxiliary graph construction is given in Figure~\ref{fig:3-tree-tsp}.

\begin{theoremrep}\label{thm:tree}
    \textsc{CCTO} is solvable in $O((nTk)^2)$ time for any instance $((G,F),s,s,k,f)$ where $G$ is a tree with $n$ vertices, the source and sink are the same vertex $s$, the lifetime of $(G,F)$ is $T$, and the maximum traversal number of any edge in $(G,F)$ is at most 3.
\end{theoremrep}

    \begin{figure}[ht]
        \centering
        \begin{tikzpicture}[every loop/.style={}, thick]
            \tikzset{vertex/.style={draw, circle, inner sep=0.55mm,fill=black}}
            \tikzstyle{diredge} = [-stealth]
        
            \node[vertex] (s) at (-1,1.5) [label={$s_{0,0}$}] {};
            
            \node[vertex] (u) at (0,0) [label={$u_{t_1,q}$}] {};
            \node[vertex] (v) at (0.5,-1) [label=below:{$v_{t_2,q+1}$}] {};
            
            \node[vertex] (u2) at (3,0) [label={$u_{t_4,q'}$}] {};
            \node[vertex] (v2) at (2.5,-1) [label=below:{$v_{t_3,q'}$}] {};
            \node[vertex] (w2) at (3.5,-2) [label=below:{$w_{t_5,q'+1}$}] {};

            \node[vertex] (u3) at (6,0) [label={$u_{t_7,q''}$}] {};
            \node[vertex] (w3) at (5.5,-2) [label=below:{$w_{t_6,q''}$}] {};

            \node[vertex] (s2) at (7,1.5) [label={$s_{T,k}$}] {};
            
            \tikzstyle{every node}=[font=\footnotesize,fill=white,inner sep=0.8pt]
            \begin{scope}[bend angle=30]
        
            \draw (s) edge[dotted] node {} (s2);
            \draw (-1,0) edge[dotted] node {} (7,0);
            \draw (-1,-1) edge[dotted] node {} (7,-1);
            \draw (-1,-2) edge[dotted] node {} (7,-2);
            
            \draw (u) edge[diredge] node {$F(u,v,t_1,t_2)$} (v);
            \draw (v2) edge[diredge] node {$F(v,u,t_3,t_4)$} (u2);
            \draw (u2) edge[diredge] node {} (w2);
            \draw (w3) edge[diredge] node {} (u3);

            \draw (s) edge[diredge,dashed,bend right] node {} (u);
            \draw (v) edge[diredge,dashed,bend right] node {} (v2);
            \draw (w2) edge[diredge,dashed,bend left] node {} (w3);

            \draw (u3) edge[diredge,dashed,bend right] node {} (s2);
            
            \end{scope}
        \end{tikzpicture}
        \caption{A weighted path in the auxiliary graph $H$ from $s_{0,0}$ to $s_{T,k}$ is equivalent to a walk from $s$ to itself in the original graph traversing at least $k$ distinct vertices in $G$. Horizontal dotted lines signify the distance of the corresponding vertices from $s$ in the original graph, dashed arcs replace segments of the walk. For simplicity, the weights are only labelled on the first two arcs pictured.}
        \label{fig:3-tree-tsp}
    \end{figure}
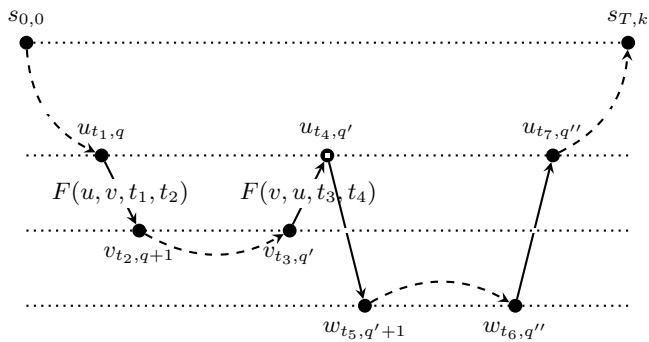

 \begin{proof}
    We note that, since the maximum traversal number of edges is at most~$3$, the source and sink are the same vertex $s$, and the underlying graph is a tree,
    an edge cannot be traversed more than twice. Else, we would not be able to return to $s$. To solve \textsc{CCTO}, we build an auxiliary graph $H$ which is similar to the weighted time expanded cost graph. Such a graph can be seen in Figure \ref{fig:3-tree-tsp}. For each $u\in V(G)$, $t\in [T]$ and $i\in [k]$ we add a vertex $u_{t,i}$. If a vertex $u$ is closer to $s$ than a vertex $v$ in $G$, we add an arc labelled $F(u,v,t_1,t_2)$ from $u_{t_1,i}$ to $v_{t_2,i'}$ where $i'=\min\{i+1,k\}$ if and only if $F(u,v,t_1,t_2)$ is finite. If $u$ is further from $s$ than $v$ is, we add an arc labelled $F(u,v,t_1,t_2)$ from $u_{t_1,i}$ to $v_{t_2,i}$
    if and only if $F(u,v,t_1,t_2)$ is finite. That is, we only increment $i$ if we are visiting a new vertex and we have not yet visited $k$ vertices.

    Since no edge can be used more than twice in a valid walk, any path in $H$ from $s_{0,0}$ to $s_{T,k}$ must correspond to a valid walk in $G$ traversing at least $k$ distinct vertices. Note that the counter of the number of vertices visited begins at $0$ so that the source vertex is not double-counted. Furthermore, since an arc is added to $H$ for any finite cost edge in $(G,F)$, there is a walk from $s$ to itself in $(G,F)$ only if there is a path from $s_{0,0}$ to $s_{T,i}$ for some $i$ in $H$. The path using least fuel from $s_{0,0}$ to $s_{T,k}$ can be calculated in $O(|V(H)|+|E(H)|)$ time by topological sorting. There are $nTk$ vertices in $H$ and thus at most $(nTk)^2$ arcs. Provided we can retrieve $F$ in constant time for any input, we can build $H$ in at most $(nTk)^2$ time. Therefore, when the underlying graph $G$ is a tree, the source and sink are the same vertex $s$, and the minimum traversal number of any edge is at most~$3$, we can solve \textsc{CCTO} in $(nTk)^2$ time.
\end{proof}

We can also generalise Theorem \ref{thm:tree} to the case where some of the edges can be allowed to have maximum traversal number greater than three.

\begin{theoremrep}\label{thm:subforest}
    Let $(G,F)$ be a temporal cost graph where $G$ is a tree with $n$ vertices, the lifetime of $(G,F)$ is $T$ and let $E'$ be a subset of edges in $G$ such that the subgraph induced by $E'$ is a forest containing at most $\ell$ leaves.
    If every edge in $E(G)\setminus E'$ has maximum traversal number~$3$, we can solve \textsc{CCTO} on any instance $((G,F),v_s,v_t,k,f)$ in $O(k^2T^2n^{2\ell+2})$ time.
\end{theoremrep}
\begin{proof}
    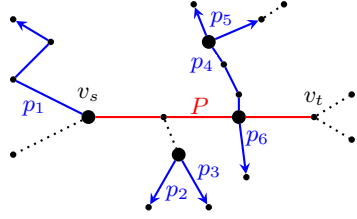
\begin{figure}[ht]
        \centering
        \begin{tikzpicture}[every loop/.style={}, thick,auto]
            \tikzset{vertex/.style={draw, circle, inner sep=0.55mm,fill=black}}
            \tikzset{smallvertex/.style={draw, circle, inner sep=0.2mm,fill=black}}
            \tikzstyle{diredge} = [-stealth]
            \tikzstyle{subforest} = [blue]
            \tikzstyle{notsubforest} = [dotted]
            \tikzstyle{mainpath} = [red]
        
            \node[vertex] (s) at (0,0) [label={$v_s$}] {};   
            \node[smallvertex] (a) at (1,0) [] {};     
            \node[vertex] (a2) at (2,0) [] {};     
            \node[smallvertex] (t) at (3,0) [label={$v_t$}] {};

            \node[smallvertex] (b) at (-1,0.5) [] {};
            \node[smallvertex] (c) at (-0.5,1) [] {};
            \node[smallvertex] (d) at (-1,1.3) [] {};
            \node[smallvertex] (e) at (-1,-0.5) [] {};

            \node[vertex] (f) at (1.2,-0.5) [] {};
            \node[smallvertex] (g) at (0.8,-1.2) [] {};
            \node[smallvertex] (h) at (1.6,-1.2) [] {};
            
            \node[smallvertex] (i) at (2,0.3) [] {};
            \node[smallvertex] (j) at (1.8,0.7) [] {};
            \node[vertex] (k) at (1.6,1) [] {};
            \node[smallvertex] (l) at (1.4,1.5) [] {};
            \node[smallvertex] (m) at (2.3,1.3) [] {};
            \node[smallvertex] (m2) at (2.6,1.5) [] {};
            
            \node[smallvertex] (n) at (2.1,-0.8) [] {};
            
            \node[smallvertex] (o) at (3.5,0.3) [] {};
            \node[smallvertex] (p) at (3.5,-0.3) [] {};

            \tikzstyle{every node}=[font=\footnotesize,fill=white,inner sep=0.8pt]
            \begin{scope}[bend angle=30]

            \draw[mainpath] (s) edge (a);
            \draw[mainpath] (a) edge node {$P$} (a2);
            \draw[mainpath] (a2) edge (t);

            \draw[subforest] (s) edge node {$p_1$} (b);
            \draw[subforest] (b) edge (c);
            \draw[subforest,diredge] (c) edge (d);
            \draw[notsubforest] (s) edge (e);
            \draw[notsubforest] (a) edge (f);
            \draw[subforest,diredge] (f) edge node {$p_2$} (g);
            \draw[subforest,diredge] (f) edge node {$p_3$} (h);
            \draw[subforest] (a2) edge (i);
            \draw[subforest] (i) edge (j);
            \draw[subforest] (j) edge node {$p_4$} (k);
            \draw[subforest,diredge] (k) edge (l);
            \draw[subforest,diredge] (k) edge node {$p_5$} (m);
            \draw[notsubforest] (m) edge (m2);
            \draw[subforest,diredge] (a2) edge node {$p_6$} (n);
            \draw[notsubforest] (t) edge (o);
            \draw[notsubforest] (t) edge (p);
            \end{scope}
        \end{tikzpicture}
        \caption{An edge-partition of a tree into paths based on the of the leaves of the subforest. Dotted edges have traversal number at most~$3$ while the rest have any traversal number.}
        \label{fig:tree-tsp}
    \end{figure}

    We solve \textsc{CCTO} by building an auxiliary graph on which we look for a path  from the source to a sink $v_{t'}$. First, consider the subforest $\mathcal{T}$ with $\ell$ leaves. We partition the edges of $\mathcal{T}$ in to $\ell$ paths $\{p_j\}_{j\in [\ell]}$ such that each path contains a leaf of $\mathcal{T}$, and every edge in $\mathcal{T}$ appears in exactly one path (see Figure~\ref{fig:tree-tsp}). We can do this in $O(n\ell)$ time by ordering the leaves arbitrarily and, for each leaf, taking the other endpoint of the path to be the closest vertex $v_{s'}\in \mathcal{T}$ to $v_s$ that is not contained in any of the other paths. For each path, we index a vertex on this path $v_0,\ldots v_p$ where $v_0=v_{s'}$ and $v_p$ is the leaf of the subtree in the path. In addition, we also find the path $P$ from $v_s$ to $v_t$ if $v_s\neq v_t$. Otherwise, $P=\{v_s\}=\{v_t\}$.
    
    We build a weighted DAG $H$ similar to the auxiliary graph used in Theorem \ref{thm:tree} and find a path from the source to the sink. The vertex set of $H$ consists of a sink vertex $v_{t'}$ and a vertex labelled $(u,t,q,q_1,\ldots,q_{\ell})$ for each $u\in V(G)\setminus \{v_s\}$, $t\in[T]$, $q\in[1,k]$, $q_j\in [n]$ for all $1\leq j\leq \ell$. Here, $u$ tells us which vertex of $G$ we are at; $t$ is a time; $q$ is the number of vertices reached by the walk from $v_s$ to $u$; and $q_j$ is $i$, the highest index of the vertices $v_0,\ldots,v_p\in p_j$ that the walk traverses for each $1\leq j \leq \ell$. If $v_s\neq v_t$, we add a source vertex labelled $(v_s,0,1,q'_1,\ldots q'_{\ell})$ where $q'_i=1$ if $v_s\in p_i$ and $q'_i=0$ otherwise. Else, if $v_s= v_t$, we add a source vertex labelled $(v_s,0,0,q'_1,\ldots q'_{\ell})$ where $q'_i=1$ if $v_s\in p_i$ and $q'_i=0$ otherwise.

    We add an arc with cost 0 from each vertex of $H$ labelled $(v_t,t,k,q_1,\ldots,q_{\ell})$ to the vertex $v_{t'}$ for all values of $t, q_1,\ldots,q_{\ell}$. After this, an arc $((u,t,q,q_1,...,q_{\ell}),(v,t',q',q_1',...,q_{\ell}'))$ with cost $F(u,v,t,t')$ is added if and only if:
    \begin{itemize}
        \item if $u = v$, $t' = t+1$, $q = q',$ and, for all $ i \in [\ell], q_i = q_i'$ (waiting on the same vertex);
        \item else, $F(u,v,t,t')<\infty$ and:
        \begin{itemize}
            \item if $v\notin \mathcal{T}$, $q_i=q_i'$ for all $1\leq i\leq \ell$, and:
            \begin{itemize}
                \item if $v$ is closer to $v_t$ than $u$, and $q' = \min\{q+1,k\}$ (moving closer to the sink and visiting an uncounted vertex);
                \item if $u$ is closer to $v_t$ than $v$ is, both $u$ and $v$ are in $P$, and $q' = q-1$ (moving away from the sink where $v$ is a vertex we must traverse again to reach the sink);
                \item if $u$ is closer to $v_t$ than $v$ is and $q' = q$ (we are moving away from the sink and the vertices will be counted on the path back to a vertex in $P$);
            \end{itemize}
            \item else, $v$ must be part of one of the $l$ paths, say~$p_j$ and:
            \begin{itemize}
                \item if the index of $v$ in $p_j$ is greater than $q_j$, $q' = \min\{q+1,k\}$, $q_j' = q_j +1$, and $q_i = q_i'$ for all $i\neq j$, $1\leq i\leq \ell$ (we are traversing a new vertex in $\mathcal{T}$);
                \item else $q' = q$ and $q_i = q_i'$ for all $i\in[\ell]$ (we are not traversing a new vertex).
            \end{itemize}
        \end{itemize}
    \end{itemize}
    We have a yes-instance of \textsc{CCTO} if and only if there exists a path in $H$ from $v_s$ to $v_{t'}$ traversing at least $k$ distinct vertices. As in Theorem \ref{thm:tree}, our argument hinges on the fact that any edge not in $\mathcal{T}$ or $P$ can be used at most twice. Similarly, any edges in $P$ and not $\mathcal{T}$ can be used either once or three times. It is clear from the construction that there exists a path from the source to $(u,t,q,q_1,...,q_{\ell})$ in $H$ if and only if there is a corresponding walk from $v_s$ to $u$ in $(G,F)$.

    We first prove that such path visits at least $k$ distinct vertices. Observe that, in order for a path from the source to reach $v_{t'}$, $q$ must be equal to $k$ in the penultimate quadruple in the path in $H$ from the source to the sink.  When the path traverses an edge in $\mathcal{T}$, $q$ is incremented only if the vertex reached has never been visited on path $p_i$, this is guaranteed by the value $q_i$ in the label of the vertex visited. For edges not in $\mathcal{T}$, $q$ is incremented only if the vertex reached is closer to $v_t$ than the previous one. If the edge is traversed more than once, there are two cases:
    \begin{itemize}
        \item if the edge is in $P$, the path between $v_s$ and $v_t$, the edge can be traversed once towards $v_t$ or thrice -- twice towards $v_t$ and once away. Since $q$ is decremented when traversing the edge away from $v_t$, $q$ can only be incremented once in total;
        \item if the edge is not in $P$, it can be used only twice -- once away from $v_t$ and once towards $v_t$. The counter $q$ is incremented in only one of these circumstances.
    \end{itemize}

    Regardless of the length of the path, $q$ will count the number of distinct vertices traversed up to $k$. A path in $H$ from the source to $v_{t'}$ represents a valid walk from $v_s$ to $v_t$ in $(G,F)$ traversing at least $k$ distinct vertices. The cost of the path is the same as the fuel cost of the walk in $(G,F)$.

    What remains is to show that we can solve \textsc{CCTO} in $O(k^2T^2n^{2\ell+2})$ time. Finding $\ell$ disjoint paths in $\mathcal{T}$ can be achieved in $O(n\ell)$ time as mentioned before. Constructing $H$ requires $kTn^{\ell+1}+1$ vertices. For each ordered pair of vertices, we can check if the arc between them exists in constant time. Thus, construction of $H$ requires $O(k^2T^2n^{2\ell+2})$ time. Finding a path from $v_s$ to $v_{t'}$ introduces an additive factor of $O(k^2T^2n^{2\ell+2})$ to the run time. This gives us the desired result.
\end{proof}
Using similar arguments, we can show that \textsc{CCTO} is solvable in polynomial time in more general graphs with a more restricted cost function. Here we do not explicitly restrict the underlying graph, but we require that there are at most 3 finite cost moves we can make from any starting point.


\begin{theoremrep}\label{thm:3-general}
    Let $(G,F)$ be a temporal cost graph with $n$ vertices and lifetime $T$ such that, for each vertex $v\in V(G)$ there are only three distinct triples $u,t_1,t_2$ such that either $F(v,u,t_1,t_2)$ or $F(u,v,t_1,t_2)$ (or both) is finite. Then, \textsc{CCTO} is solvable on any instance $((G,F),v_s,v_t,k,f)$ in $O((nTk)^2)$ time.
\end{theoremrep}

\begin{proof}
    As in Theorem \ref{thm:tree}, we build the graph $H$ whose vertex set consists of a vertex for each $u\in V(G)$, $t\in [T]$ and $i\in [k]$. We add an arc labelled $F(u,v,t_1,t_2)$ from $u_{t_1,i}$ to $v_{t_2,i'}$ where $i'=\min\{i+1,k\}$ if and only if either $F(v,u,t_1,t_2)$ or $F(u,v,t_1,t_2)$ is finite. 

    Since, for each vertex $v\in V(G)$ there are only three distinct triples $u,t_1,t_2$ such that either $F(v,u,t_1,t_2)$ or $F(u,v,t_1,t_2)$ is finite, each vertex can only be traversed once (unless it is both the source and sink vertex). Denote the source and sink vertices by $s$ and $t$ respectively. If $s=t$, we consider paths starting from $s_{0,0}$. Else, we consider paths starting from $s_{0,1}$.  A path from $s_{0,1}$ to $v_{t',i}$ exists in $H$ if and only if there exists a corresponding walk in $(G,F)$ from $s$ to $v$ arriving by $t'$ and traversing $i$ distinct vertices. Thus, if $s=t$, we have a yes-instance if and only if there is a path in $H$ from $s_{0,0}$ to $t_{T,k}$ where $T$ is the lifetime of $(G,F)$. If $s\neq t$, we have a yes-instance if and only if there is a path in $H$ from $s_{0,1}$ to $t_{T,k}$ where $T$ is the lifetime of $(G,F)$.

    The path using least fuel from $s_{0,0}$ or $s_{0,1}$ to $s_{T,k}$ can be calculated in $O(|V(H)|+|E(H)|)$ time by topological sorting. There are $nTk$ vertices in $H$ and thus at most $(nTk)^2$ arcs. Provided we can retrieve $F$ in constant time for any input, we can build $H$ in at most $(nTk)^2$ time. Therefore, if for each vertex $v\in V(G)$ there are only three distinct triples $u,t_1,t_2$ such that either $F(v,u,t_1,t_2)$ or $F(u,v,t_1,t_2)$ is finite, we can solve \textsc{CCTO} in $(nTk)^2$ time.
    \end{proof}

\section{Vertex-interval-traversal number}\label{sec:vitw}

Bumpus and Meeks \cite{bumpus2021edge} introduce the parameter \emph{vertex-interval-membership width}, $\omega(G,\lambda)$ of a temporal graph $(G,\lambda)$, and give an algorithm which computes $\omega(G,\lambda)$ in time $O(\omega(G,\lambda)\tau)$ where $\tau$ is the latest time an edge in $(G,\lambda)$ is active. The vertex-interval-membership width is calculated by finding the \emph{vertex-interval-membership sequence} of the temporal graph $(G,\lambda)$. This is the sequence $(F_t)_{t\in[\tau]}$ of vertex-subsets of $(G,\lambda)$ where $F_t:=\{v\in V(G):\exists i\leq t\leq j\text{ and }u,w\in V(G)\text{ such that }i\in \lambda(uv) \text{ and }j\in\lambda(wv)\}$. The vertex-interval-membership width is $\omega((G,\lambda)):=\max_{t\in \tau}|F_t|$.  Roughly speaking, for many problems this parameter bounds the number of vertices of interest at any time (where, for example, a vertex may be considered interesting at time $t$ if we could have reached it before $t$ and subsequently depart at some time after $t$)~\cite{enright2025families}.

We generalise the definition of vertex-interval-membership sequence to define an analogous parameter for temporal cost graphs.

\begin{definition}[Vertex-interval-traversal width] 
The \emph{vertex-interval-traversal sequence} of a temporal cost graph $(G,F)$ with lifetime $T$ is the sequence $(H_t)_{t\in[T]}$ of vertex-subsets of $G$ where $H_t:=\{v\in V(G):\exists i\leq t\leq j, t_1, t_2\in[T]\text{ and }u,w\in V(G)\text{ such that }F(u,v,t_1,i)<\infty \text{ and }F(v,w,j,t_2)<\infty\}$. Note that $u$ and $w$ may be the same vertex. The \emph{vertex-interval-traversal width} is $\varphi((G,F)):=\max_{t\in [T]}|H_t|$.
\end{definition}

We begin by observing that we can compute the vertex-interval-traversal sequence efficiently.

\begin{lemmarep}\label{lem:vitw-computing-time}
    Let $G$ be a graph with cost function $F$ with lifetime $T$ such that every edge of $G$ has finite cost for at least one pair of times. Then we can find the vertex-interval-traversal sequence of $(G,F)$ in time $O(n^2T^2)$.
\end{lemmarep}

\begin{proof}
    We begin by initialising a sequence $(H_t)_{t\in[T]}$ of empty sets. For each vertex $v\in V(G)$, we find
    \begin{itemize}
        \item $t_{\min}(v)$ the minimum time $t$ such that there exists a vertex $v_2$ and time $t_2$ such that either $F(v,v_2,t,t_2)$ or $F(v_2,v,t_2,t)$ is finite, and
        \item $t_{\max}(v)$ the maximum time $t$ such that there exists a vertex $v_3$ and time $t_3$ such that either $F(v,v_3,t,t_3)$ or $F(v_3,v,t_3,t)$ is finite.
    \end{itemize} 
    For all times $t_{\min}(v)\leq t'\leq t_{\max}(v)$, we add $v$ to $H_{t'}$.

    If we assume we can find $F$ for any inputs in constant time, we must find $F$ at most $2n^2T^2$ times to calculate $t_{\min}$ and $t_{\max}$ for each vertex in $G$. Following this, we add $v$ to $H_{t'}$ for every $t_{\min}(v)\leq t'\leq t_{\max}(v)$. This takes $O(n\varphi)\leq O(n^2)$ time. Therefore, the algorithm takes at most $2n^2T^2+n\varphi\leq 2n^2T^2+n^2=O(n^2T^2)$ time.
\end{proof}

Equipped with this notion, we give the following result. Our algorithm functions by dynamic programming over the vertex-interval-traversal sequence. At each time $t$, we find a set of states corresponding to valid walks which record the last vertex on the walk, the number of vertices traversed on the walk before $t$ which are not in the bag $H_t$, the set of vertices in $H_t$ traversed by the walk, and the fuel cost of the walk so far.

\begin{theoremrep}\label{thm:vitw}
    Given a temporal cost graph $(G,F)$ with $n$ vertices, lifetime $T$, and vertex-interval-traversal width $\varphi$, we can solve \textsc{CCTO} in $O(2^{2\varphi}T(\varphi  kf)^2+n^2T^2)$ time, or $O(2^{2\varphi}T(\varphi  kf)^2)$ time when the vertex-interval-traversal sequence is given.
\end{theoremrep}

\begin{proof}
    We show this by dynamic programming over the vertex-interval-traversal sequence. As in Lemma \ref{lem:vitw-computing-time}, we can find this in $O(n^2T^2)$ time. For each set $H_i$, we compute a set of states $(v_{\text{fin}},\ell_{\text{before}},\ell_{\text{in}},f^*)$ where $v_{\text{fin}}\in \{v\in H_i\}$, $\ell_{\text{before}}\in [0,k]$, $\ell_{\text{in}}\subseteq H_i$, and $f^*\in [0,f]$. We call a state $(v_{\text{fin}},\ell_{\text{before}},\ell_{\text{in}},f^*)$ of $H_i$ \emph{valid} if and only if there is a valid walk from $v_s$ to $v_{\text{fin}}$ that traverses $\ell_{\text{before}}$ vertices in $\bigcup_{j<i}H_j\setminus H_i$ and the vertices in $H_i$ in accordance to $\ell_{\text{in}}$ using $f^*$ fuel. We have a yes-instance of \textsc{CCTO} if there is a valid state in the final set of the sequence such that $v_{\text{fin}}=v_t$, the fuel used is at most $f$, and the number of vertices in the walk described by that state is at least $k$.

    We assume without loss of generality that there exists a finite-cost edge departing $v_s$ at time 1 and another finite-cost edge which arrives at $v_t$ at time $T$, the lifetime of the temporal cost graph. Else, we apply a uniform shift to the times and disregard all times after the last finite cost tuple arrives at $v_t$. We also create the set $H_0$ to which we add only $v_s$. 
    For a given set $H_j$, there are at most $\varphi\cdot k\cdot 2^{\varphi}\cdot f$ valid states.

    The set $H_0$ is given the trivial state $(s,0,\ell_{\text{in}}=\{v_s\}, 0)$. Given a state $(v_{\text{fin}},\ell_{\text{before}},\ell_{\text{in}},f^*)$ of $H_{j}$, we say it is valid if and only if 
    \begin{itemize}
        \item there is a state $(v_{\text{fin}},\ell'_{\text{before}},\ell'_{\text{in}},f^*)$ of $H_{j-1}$ where $\ell_{\text{before}}=\min\{k,\ell'_{\text{before}}+|\ell'_{\text{in}}\setminus H_{j}|\}$ (we do not move vertices at time $j$); or
        \item there is a state $(v_{\text{fin}}'',\ell_{\text{before}}'',\ell_{\text{in}}'',f'')$ of $H_{j''}$ with $j''<j$ (we move vertices) such that:
        \begin{itemize}
            \item $f''+F(v''_{\text{fin}},v_{\text{fin}},j'',j)<f$ (the fuel cost is still below $f$);
            \item $|\ell_{\text{in}}''\cap (H_{j''}\setminus H_j)|\geq \ell_{\text{before}}-\ell_{\text{before}}''$ (any vertices forgotten from $H_{j''}$ to $H_j$ which are traversed by the walk are counted);
            \item $\ell_{\text{in}}= (\ell''_{\text{in}}\cap H_j)\cup \{v_{\text{fin}}\}$ (the final vertex in the walk is marked as reached);
            \item for all $v\in H_{j''}\setminus H_{j'}$, $v\in \ell_{\text{in}}$ if and only if $v=v_{\text{fin}}$ (any vertices introduced between $H_{j''}$ and $H_j$ are only reached if they are the final vertex on the walk);
            \item $f^*=f''+F(v''_{\text{fin}}, v_{\text{fin}}, j'',j)$ (the fuel used to traverse these vertices by time $j$ is minimised).
        \end{itemize}
    \end{itemize}
    We now show correctness of the dynamic program. We aim to show that a state is \emph{valid} if and only if it is realised by a walk.
    \begin{claim}
        There is a valid state $(v_{\text{fin}},\ell_{\text{before}},\ell_{\text{in}},f^*)$ of a set $H_j$ if and only if there is a walk from $v_s$ to $v_{\text{fin}}$ that uses fuel $f^*$ which traverses the vertices in $\ell_{\text{in}}$ and at least $\ell_{\text{before}}$ vertices whose incident edges only have finite cost at times strictly before $j$.
    \end{claim}
    \begin{nestedproof}
        We show the claim by induction on $j$. Our base case is the set $H_0$ which consists of only $v_s$ and has the state $(s,0,\{v_s\}, 0)$ which is realised by the trivial walk from $v_s$ to itself. 

    Now, assume that $j>0$ and for all $j'\leq j$, there is a valid state $(v_{\text{fin}},\ell_{\text{before}},\ell_{\text{in}},f^*)$ of $H_{j'}$ if and only if it is realised by a walk. That is, there is a valid walk from $v_s$ to $v_{\text{fin}}$ arriving by time $j'$ that uses $f^*$ fuel, traverses the vertices $v$ in the set $\ell_{\text{in}}\subseteq H_{j'}$ and at least $\ell_{\text{before}}$ vertices using edges at times strictly before $j'$. 
    
    We begin by assuming we have a valid state $(v_{\text{fin}},\ell_{\text{before}},\ell_{\text{in}},f^*)$ of $H_{j+1}$. If it meets the first criterion and we have not moved vertices at time $j$, then there is a valid state $(v_{\text{fin}},\ell'_{\text{before}},\ell'_{\text{in}},f^*)$ of $H_j$ where $\ell_{\text{before}}=\min\{k,\ell'_{\text{before}}+|\ell'_{\text{in}}\setminus H_{j}|\}$. The difference in $\ell_{\text{before}}$ and $\ell_{\text{in}}$ account for the vertices forgotten between $H_j$ and $H_{j+1}$. Therefore, if there is a walk that realises the state of $H_j$ it must also realise $(v_{\text{fin}},\ell_{\text{before}},\ell_{\text{in}},f^*)$. We have assumed this to be true by the inductive hypothesis; thus $(v_{\text{fin}},\ell_{\text{before}},\ell_{\text{in}},f^*)$ is valid.
    
    If, for the state $(v_{\text{fin}},\ell_{\text{before}},\ell_{\text{in}},f^*)$ of $H_{j+1}$, there is a valid state $(v_{\text{fin}}'',\ell_{\text{before}}'',\ell_{\text{in}}'',f'')$ of $H_{j''}$ for some $j''\leq j$ and an edge between $v_{\text{fin}}$ and $v_{\text{fin}}''$ such that $f''+F(v''_{\text{fin}},v_{\text{fin}},j'',j+1)<f$, the walk found by adding the edge $v_{\text{fin}}''v_{\text{fin}}$ departing at time $j''$ and arriving at time $j+1$ to the walk which realises $(v_{\text{fin}}'',\ell_{\text{before}}'',\ell_{\text{in}}'',f'')$ must be a valid walk. The walk must traverse $\ell_{\text{before}}$ vertices which are only appear in sets $H_{j^*}$ for $j^*<j$ where $\ell_{\text{before}}$ is the sum of $\ell_{\text{before}}''$ and the vertices traversed by the walk which are forgotten between $H_{j'}$ and $H_{j+1}$. Of the vertices in $H_{j+1}$ the walk must also traverse only the vertices in $\ell_{\text{in}}''$ and $v_{\text{fin}}$. Finally, the fuel used in this walk must be $f^*=f''+F(v_{\text{fin}}'', v_{\text{fin}},j'',j+1)$.

    Now suppose there is a walk $W$ from $v_s$ to $v_{\text{fin}}$ over $\ell_{\text{before}}$ vertices which have been forgotten by time $j+1$ and the vertices in $\ell_{\text{in}}$ that uses fuel $f^*$. If $v_{\text{fin}}$ is reached by the walk by time $j$, there must be a state $(v_{\text{fin}},\ell'_{\text{before}},\ell'_{\text{in}},f^*)$ of $H_j$ which is realised by $W$ where $\ell_{\text{before}}=\min\{k,\ell'_{\text{before}}+|\ell'_{\text{in}}\setminus H_{j}|\}$. Therefore, $(v_{\text{fin}},\ell_{\text{before}},\ell_{\text{in}},f^*)$ is a valid state of $H_{j+1}$.
    
    Suppose that $W$ arrives at $v_{\text{fin}}$ at time $j+1$ exactly. Then, let $v'_{\text{fin}}$ be the penultimate vertex on the walk. Note that this must exist as, if $v_{\text{fin}}$ was the source vertex $v_s$, $W$ would reach $v_{\text{fin}}$ trivially at time 0. Let $j'$ be the time at which the walk $W$ departs the vertex $v'_{\text{fin}}$ and let $W'$ be the walk found by removing the final edge from $W$. By the inductive hypothesis, there is a state $(v''_{\text{fin}},\ell''_{\text{before}},\ell''_{\text{in}},f^*-F(v''_{\text{fin}}, v_{\text{fin}},j'' j+1))$ of $H_{j''}$ which is realised by $W'$. In addition, $|\ell_{\text{in}}''\cap (H_{j''}\setminus H_j)|\geq \ell_{\text{before}}-\ell_{\text{before}}''$ and $\ell_{\text{in}}= (\ell''_{\text{in}}\cap H_j)\cup \{v_{\text{fin}}\}$. Hence, the state $(v_{\text{fin}},\ell_{\text{before}},\ell_{\text{in}},f^*)$ must be supported by $W$.
    \end{nestedproof}
We have a yes-instance of \textsc{CCTO} if and only if there is a valid state with $v_{\text{fin}}=v_t$, $f^*\leq f$ and $|\ell_{\text{in}}|+\ell_{\text{before}}\geq k$. As stated in Lemma \ref{lem:vitw-computing-time}, we can find the value $\varphi$ in $O(n^2T^2)$ time. Given a pair of states such that one is a valid state of an earlier set in the vertex-interval-traversal sequence, we can perform the checks required to validate the later state in constant time. We have at most $T$ sets of the vertex-interval-traversal sequence to determine the valid states of and at most $\varphi\cdot k\cdot 2^{\varphi}\cdot f$ states per set. This gives us a runtime of $O(2^{2\varphi}T(\varphi  kf)^2+n^2T^2)$.
\end{proof}

\section{Target number of objects}\label{sec:k-small}
In this section we use a colour-coding argument to show that \textsc{CCTO} is tractable with respect to the number of vertices we are required to visit: 

\begin{theorem}\label{thm:fpt_k}
    Given a temporal cost graph $(G,F)$ with $n$ vertices and lifetime $T$, and a natural number $k$ indicating the target number of vertices to visit, \textsc{CCTO} is solvable in time $k^{O(k)}T^4n^4\log n$.
\end{theorem}

 Intuitively, using a given arbitrary colouring of the input graph with $k$ colours and an ordering of those colours, we give an algorithm which efficiently verifies whether there exists a valid walk of cost at most $f$ in which at least one vertex of each colour is visited, and the order in which the first vertex of each colour is visited matches the specified order.  In particular, this ensures that the walk traverses at least $k$ distinct vertices (although for some colourings and orderings our algorithm may well find walks traversing more than $k$ distinct vertices).  
 Pairing this algorithm with a standard colour-coding technique solves \textsc{CCTO}. Our algorithm makes use of two auxiliary cost functions. The first gives us the minimum cost of a valid walk from one vertex to another given the departure and arrival times of the walk.

\begin{lemmarep}\label{lem:non-direct-travel}
    Let $(G=(V,E),F)$ be a temporal cost graph with $n$ vertices and lifetime $T$.
    We can compute, in time $O(n^3T^4)$, a function $\hat{F}: V \times V \times [T] \times [T] \rightarrow \mathbb{N}$ such that, for each pair of (not necessarily adjacent) vertices $(u,v)$ and each pair of times $(t_1,t_2)$, we have that $\hat{F}(u,v,t_1,t_2)$ is the minimum cost of any valid walk from $u$ to $v$ departing at time $t_1$ and arriving at time $t_2$ (or $\infty$ if no such walk exists).
\end{lemmarep}
\begin{proof}
    We will consider each source vertex $u$ in turn, and give a dynamic program which computes $\hat{F}(u,v,t_1,t_2)$ for all $v$, $t_1$ and $t_2$.  We compute the values for each arrival time in turn, starting with $0$: we initialise $\hat{F}(u,u,0,0) = 0$ and $\hat{F}(u,v,t,0) = \infty$ for all $v \neq u$ (since we cannot arrive at any vertex other than the source by time $0$).  We calculate values for arrival times $i > 0$, assuming we have already computed all values for arrival times $j < i$, via the following claim. When $v=u$, $\hat{F}(u,v,t_1,i)=0$ for all values of $t_1, i$ such that $t_1\leq i$. If $t_1=i$ and $v\neq u$, $\hat{F}(u,v,t_1,i)=\infty$ by our requirement that all edges in a temporal cost graph have positive cost.
    \begin{claim}
    For every vertex $v\neq u$, and all times $i$ with $1 \le t_1 < i \le T$ we have
        $$\hat{F}(u,v,t_1,i) = \min \{\hat{F}(u,v,t_1,t')+F(v',v,t'',i): t_1 \le t'\leq t''<i \text{ and } vv' \in E \}.$$
    \end{claim}
    \begin{nestedproof}
        We first argue that, for all $t' \leq t'' \leq i$ and $v'$ with $vv' \in E$ we have $\hat{F}(u,v,t_1,i) \le \hat{F}(u,v,t_1,t')+F(v',v,t'',i)$.  Assume without loss of generality that the right-hand side is finite. Then, by definition of $\hat{F}(u,v,t_1,t')$, we know that there exists a valid walk $W$ from $u$ to $v$, departing at time $t_1$ and arriving at time $t'$, with cost $\hat{F}(u,v,t_1,t')$.  Moreover, we know that we can travel from $v'$ to $v$ via a single edge, departing at time $t''$ and arriving at time $i$, with cost $F(v',v,t'',i)$.  Since $t''>t'$, we can append this edge to $W$, giving a valid walk from $u$ to $v$, departing at time $t_1$ and arriving at time $i$, with cost $\hat{F}(u,v,t_1,t')+F(v',v,t'',i)$.  Thus, as $\hat{F}(u,v,t_1,i)$ is defined to be the minimum cost of such a walk, the inequality follows.

        To complete the proof of the claim we now show that there exist $t' \leq t'' \leq i$ and $vv' \in E$ such that $\hat{F}(u,v,t_1,i) \ge \hat{F}(u,v,t_1,t')+F(v',v,t'',i)$.  By definition, there exists a valid walk $W$ from $u$ to $v$, departing at time $t_1$ and arriving at time $i$, with cost $\hat{F}(u,v,t_1,i)$.  Since $i>1$, and the walk arrives at time exactly $i$, this walk traverses at least one edge.  Let $(v',v',t'',i)$ be the final element of $W$.  Removing this final element must give a shorter valid walk $W'$ from $u$ to $v'$, departing at time $t_1$ and arriving at some time $t'$ with $t_1 < t' \le t''$.  By definition, the cost of $W'$ is at least $\hat{F}(u,v,t_1,t')$, so the cost of $W$ is at least $\hat{F}(u,v,t_1,t')+F(v',v,t'',i)$, as required.
    \end{nestedproof}
    It remains to bound the time taken to carry out these computations.
    \begin{claim}
        The dynamic program returns $\hat{F}(u,v,t_1,t_2)$ for all $u,v\in V(G)$ and $t_1$, $t_2\in [T]$ in $O(n^3T^4)$ time.
    \end{claim}
    
    \begin{nestedproof}
        We have $V\times [T]\times [T]$ values of $\hat{F}$ to calculate for each source. We compute each of these by taking the minimum over a set of at most $T^2n$ values, where each is calculated in constant time by adding pairs of values (where each is either already computed or provided as part of the input to the problem).  The algorithm therefore requires time $O(n^2 T^4)$ for each source, giving an overall running time of $O(n^3 T^4)$.

        This concludes the proof of the Claim.
    \end{nestedproof}
    With this final Claim proven, this then gives us the overall Lemma result. 
\end{proof}

Using Lemma \ref{lem:non-direct-travel}, we now show how to solve a variant of \textsc{CCTO} in which objects have colours, and there are restrictions on the order in which we must visit objects of different colours.  
The colourful version of our problem is defined as follows:
\begin{framed}
\noindent
\textbf{\textsc{Colourful Changing Cost Temporal Orienteering \\(CCCTO)}}\\
\textit{Input:} A temporal cost graph $(G=(V,E),F)$, a source vertex $v_s$, sink vertex $v_t$, a target value $k \in \mathbb{N}$, a colouring $c:V \setminus \{v_s,v_t\} \rightarrow [k-2]$,  and budget $f \in \mathbb{N}$.\\
\textit{Question:} Is there a valid walk in $(G,F)$ from $v_s$ to $v_t$ that visits at least one vertex of each colour and has cost at most $f$?
\end{framed}

The idea of colour coding is to repeatedly solve the colourful version of the problem on instances whose colourings are drawn from an $(n,k)$-perfect hash family.
\begin{definition}[$(n,k)$-perfect hash family]
    An $(n,k)$\emph{-perfect hash family} $\mathcal{F}$ is a family of functions from $[n]$ to $[k]$ such that for every set $S$ of size $k$ there exists a function $c$ in $\mathcal{F}$ that is injective on $S$.
\end{definition}

In general, the colour-coding method works as follows.  If we know there is a solution to some problem that contains exactly $k$ vertices, we can be sure there is at least one colouring in the family under which these $k$ vertices all receive different colours, and so this colouring will give a yes-instance for the colourful problem.  This is exactly the approach used to solve problems such as $k$-\textsc{Path}~\cite{alon_colour_1995}.  

Here there is subtle but technically important difference from this standard application of the method, because the colourful version of our problem looks for solutions using \emph{at least} one vertex of each colour, not exactly one.  This is necessary for two reasons.  First, as we have a specified sink vertex, it may be that the least costly walk from source to sink that visits at least $k$ vertices in fact visits more than $k$ vertices (and we cannot simply truncate after the first $k$).  More importantly, as we are looking for walks rather than paths, we may wish to return to previously visited vertices on the way to a new one; since we want to avoid keeping track of all vertices previously visited (this would require us to consider $n^{\Theta(k)}$ possibilities in the worst case) we cannot distinguish between revisiting a vertex we have already counted and visiting a new vertex with the same colour, so we must allow the inclusion of multiple vertices with the same colour.

We will in fact consider all possible orders in which colours might \emph{first} appear on an optimal walk meeting our requirements.  Given a permutation $\pi$ of the set $C$ of colours, we say that a valid walk $W$ \emph{respects $\pi$} if the order in which colours first appear on $W$ matches the order given by $\pi$. Our second auxiliary cost function gives us the minimum cost of a walk from the source to a vertex which arrives by a specified time and respects a specified colouring. We then solve \textsc{CCCTO} by iterating over all permutations $\pi$ and using the algorithm given in the following lemma.

\begin{lemmarep}\label{lem:ordered-path-min}
    Given a temporal cost graph $(G=(V,E),F)$ with $n$ vertices and lifetime $T$, a vertex $v_s \in V$, a colouring $c:V \rightarrow [k] \cup \{0\}$ with $c^{-1}(0) = \{v_s\}$ and a permutation $\pi$ of $[k]$, we can find in time $O(kn^4T^4)$ the minimum cost of a valid walk that respects $\pi$, the permutation of $[k] \cup \{0\}$ obtained from $\pi$ by adding $0$ at the start. 
\end{lemmarep}
\begin{proof}
    We use a dynamic programming approach.  Before describing this approach, we need to introduce some notation.

    Throughout, we will denote the $i$th colour under $\pi$ by $c_i$; we extend this to include colour 0 by setting $c_0 = 0$.   For each $i \in [k]$, let $V_i \subseteq V$ be the set of vertices with $c(v) = i$, and for each $i,j \in [k]$ let $E_{i,j}$ be the set of edges with one endpoint of colour $i$ and one of colour $j$.  We denote by $T$ the lifetime of $(G,F)$.  Note that any valid walk respecting $\pi$ must start at $v_s$, since by assumption this is the only vertex with colour $0$.

    Our dynamic program works through colours in turn, in the order specified by $\pi$.  For each $i$ and for each $t \in T$, we will compute, for all $v \in V_i$, the minimum cost $F^*_\pi(v,t)$ of a walk from $v_s$ to $v$ that arrives by time $t$ and respects $\pi_i = c_0 c_1 c_2 \dots c_i$, with $v$ the first vertex of colour $c_i$ on the walk; we will assume we have already computed these costs for every $j<i$.  It is clear that computing these costs for $i=k$ will solve the problem (we take the minimum over all vertices in $V_k$).


    We begin by initialising $F^*_\pi(v_s,t) = 0$ for all $t$ (since the only walk satisfying the requirements is the trivial walk, which by convention has no cost).

    To compute further values of $F^*_\pi$, we will use the algorithm of Lemma \ref{lem:non-direct-travel} as a subroutine.  This allows us to compute, for any vertices $u$ and $v$, times $t_1$ and $t_2$, and set $I \subseteq [k] \cup \{0\}$, the minimum cost of a valid walk departing from $u$ at time $t_1$ and arriving at $v$ at time $t_2$, such that all vertices on the walk other than possibly $v$ receive colours from $I$ under $c$; we denote this quantity by $\hat{F}_I(u,v,t_1,t_2)$.  To compute $\hat{F}_I(u,v,t_1,t_2)$ we invoke the algorithm of Lemma \ref{lem:non-direct-travel} on the temporal cost graph induced by $\{u,v\} \cup \bigcup_{i \in I}V_i$.

    Before proceeding with our dynamic program, we precompute the values of $\hat{F}_I(u,v,t_1,t_2)$ that we may need.  Specifically, for each $u,v \in V$ and $t_1,t_2 \in [T]$, if $c(u) = c_i$ and $c(v)=c_j$ with $i<j$, we set $I = \{c_0,c_1,\dots,c_{j-1}\}$ (that is, the set of all colours appearing in $\pi$ strictly before $c(v)$), and compute and store the value $\hat{F}_I(u,v,t_1,t_2)$.  Notice that we use at most $k$ different induced subgraphs of our temporal cost graph, each of which can be constructed in time $O(n)$.  Using the algorithm of Lemma \ref{lem:non-direct-travel}, we can compute the minimum cost walk for all tuples of vertices and times for each induced subgraph in time $O(n^3T^4)$.  Overall, therefore, this precomputation step requires time $O(kn^3T^4)$.


    Equipped with these values, our dynamic program uses the values of $F^*_\pi(v,t)$ for $v \in V_{i-1}$ to compute those for $v \in V_i$.  


    Given $F^*_{\pi}$ for each vertex in $V_{i-1}$ and time $t$, we calculate $F^*_{\pi}$ for each vertex-time pair $(v_i,t)$ with $v_i \in V_i$ using the following relationship. 

    \begin{claim}
        Fix $i \geq 1$.  For all $v\in V_i$ with $t\in [T]$, 
        $$F^*_{\pi}(v_i,t) = \min \{ F^*_{\pi}(v_{i-1},t_1)+\hat{F}_{\{c_0,\dots,c_{i-1}\}}(v_{i-1},v_i,t_1,t_2): v_{i-1}\in V_{i-1}\text{ and } 1 \le t_1 < t_2\leq t \}.$$
    \end{claim}
    \begin{nestedproof}

        We begin by showing that 
        $$F^*_{\pi}(v_i,t) \le F^*_{\pi}(v_{i-1},t_1)+\hat{F}_{\{c_0,\dots,c_{i-1}\}}(v_{i-1},v_i,t_1,t_2)$$
        for every $v_{i-1}\in V_{i-1}$ and $1 \le t_1 < t_2\leq t$.
        We assume without loss of generality that the right-hand side is finite.  By definition, since $F^*_{\pi}(v_{i-1},t_1) < \infty$, there is a valid walk from $v_s$ to $v_{i-1}$ that respects $\pi_{i-1}$ and arrives by time $t$ and has cost $F^*_{\pi}(v_{i-1},t_1)$.  Similarly, it follows from the definition and the fact that $\hat{F}_{\{c_0,\dots,c_{i-1}\}}(v_{i-1},v_i,t_1,t_2) < \infty$ that there is a valid walk from $v_{i-1}$ departing at time $t_1$ and arriving at $t_2 < t$ and using only vertices with colours in $\{c_0,\dots,c_{i-1}\}$, with cost $\hat{F}_{\{c_0,\dots,c_{i-1}\}}(v_{i-1},v_i,t_1,t_2)$.  Since the first of these walks reaches $v_{i-1}$ no later than time $t_1$ and the second departs at time $t_1$, they can be concatenated to give a valid walk from that arrives by time $t$ and has cost $F^*_{\pi}(v_{i-1},t_1)+\hat{F}_{\{c_0,\dots,c_{i-1}\}}(v_{i-1},v_i,t_1,t_2)$.  Moreover, all colours $c_0,\cdots,c_{i-1}$ appear in the first part of the path in the order given by $\pi$, and $v_i$ is the first vertex on the walk with colour $c_i$, so the walk respects $\pi_i$.  Since $F^*_{\pi}(v_i,t)$ is defined to be the minimum cost of a path with these properties, the inequality follows.

        We now complete the proof of the claim by showing that there exist $v_{i-1}' \in V_i$ and $1 \le t_1' < t_2' \le t$ such that 
        $$F^*_{\pi}(v_i,t) \ge F^*_{\pi}(v_{i-1}',t_1')+\hat{F}_{\{c_0,\dots,c_{i-1}\}}(v_{i-1}',v_i,t_1',t_2').$$
        By definition, there exists a valid walk $W'$ from $v_s$ to $v_i$ arriving by time $t$ which respects $\pi$ and has cost $F^*_{\pi}(v_i,t)$.  Since $v_i \neq v_s$ (as $i \ge 1$), the walk $W'$ most contain at least one edge.  As $W'$ respects $\pi$, it must contain some vertex $v_{i-1}' \neq v_i$ of colour $c_{i-1}$; assume without loss of generality that $v_{i-1}'$ is the first vertex of colour $c_{i-1}$ on the walk.  Let $t_1'$ be the time at which $W'$ first departs from $v_{i-1}'$, and $t_2'$ the time at which $W'$ reaches $v_i$.  We can therefore split $W'$ into two shorter valid walks $W_1'$ and $W_2'$, where:
        \begin{itemize}
            \item $W_1'$ is a valid walk from $v_s$ to $v_{i-1}'$ that arrives by time $t_1'$, respects $\pi_{i-1}$, and has the property that $v_{i-1}'$ is the first vertex of colour $c_{i-1}$ on the walk, and
            \item $W_2'$ is a valid walk from $v_{i-1}'$ to $v$ departing at time $t_1'$ and arriving at $t_2'$, that, except for $v$, uses only vertices with colours from the set $\{c_0,\dots,c_{i-1}\}$.
        \end{itemize}
        It follows that $W_1'$ has cost at least $F^*_\pi(v_{i-1}',t_1')$ and $W_2'$ has cost at least $\hat{F}_{\{c_0,\dots,c_{i-1}\}}(v_{i-1}',v_i,t_1',t_2')$; since $W'$ is the concatenation of these two walks, we see that $W'$ has cost at least $F^*_{\pi}(v_{i-1}',t_1')+\hat{F}_{\{c_0,\dots,c_{i-1}\}}(v_{i-1}',v_i,t_1',t_2')$ as required.

    \end{nestedproof}
    To complete the proof of the lemma, it remains only to bound the running time of the dynamic programming algorithm.
    \begin{claim}
        The algorithm computes $F^*_{\pi}(v,t)$ for all $v \in V$ and $t \in [T]$ in time $O(kn^3T^4)$. 
    \end{claim}
    \begin{nestedproof}
        Recall that we begin with a precomputation step in which we compute values of $\hat{F}_I$, which takes time $O(kn^3T^4)$.  Initialising the values of $F^*_\pi$ for $v=v_s$ takes time $O(T)$.  We then consider each $i$ from $1$ to $k$ in turn, and calculate $F^*_\pi(v,t)$ for each $v \in V_i$ and $t \in [T]$.  For each combination of $v$ and $t$ (of which there are at most $nT$) we need to take the minimum over at most $nT^2$ values (considering all possibilities for $v_{i-1}$, $t_1$ and $t_2$), where each of these values is computed in constant time by adding previously computed values.  The total time required by the algorithm is therefore $O(kn^3T^4) + O(n^2T^3) = O(kn^3T^4)$.
    \end{nestedproof}
    Since we have shown that $F^*_{\pi}$ gives the minimum cost of a walk which respects the ordering $\pi$, $F^*_{\pi}(v_T,T)$ where $v_T$ is the sink must be the minimum cost of a walk traversing at least $k$ vertices in the order given by $\pi$.
\end{proof}
   
We now use Lemma \ref{lem:ordered-path-min} to solve the colourful version of our problem by trying each possible order in which the colours can first appear on the walk.

\begin{theoremrep}\label{thm:colour-solve}
    Given a temporal cost graph $(G,F)$ with $n$ vertices and lifetime $T$, we can solve \textsc{CCCTO} on $(G,F)$ in $O(k!n^4T^4)$ time.
\end{theoremrep}
\begin{proof}
    Let $\left((G=(V,E),F),v_s,v_t,k,c,f\right)$ be the input to an instance of \textsc{CCCTO}.  We begin by extending the colouring $c$ to $v_s$ and $v_t$ by setting $c(v_s) = 0$ and $c(v_t) = k-1$.  It is clear that there is a valid walk in $(G,F)$ from $v_s$ to $v_t$ that visits at least one vertex of each colour in $[k-2]$ and has cost at most $f$ if and only if there exists a permutation $\pi'$ of $[k-2]$ such that, if $\pi$ is the permutation of $[k-1] \cup \{0\}$ obtained from $\pi'$ by inserting $0$ at the beginning and $k-2$ at the end, there is a valid walk in $(G,F)$ of cost at most $f$ that respects $\pi$.

    We therefore solve the problem by iterating over all possibilities for $\pi'$ (of which there are exactly $(k-2)!$), construct the permutation $\pi_2$ of $[k-1]$ by appending $k-1$, and run the algorithm of Lemma \ref{lem:ordered-path-min} on input $((G,F),v_s,c,\pi_2)$.  If the result returned for any choice of $\pi'$ is at most $f$ then we conclude we have a yes-instance; otherwise we conclude that we have a no-instance.

    The running time bound follows immediately from Lemma \ref{lem:ordered-path-min}.
\end{proof}

We now show how to use our algorithm for \textsc{CCCTO} to solve the uncoloured version of the problem.  We will make use of the following result about $(n,k)$-perfect families of hash functions, due to Naor et al. \cite{naor95}.
\begin{theorem}[\cite{naor95}]\label{thm:hash-family}
    For any $n$, $k\geq 1$, an $(n,k)$-perfect hash family of size $e^k k^{O(\log k)}\log n$ can be found in time $e^k k^{O(\log k)}n\log n$.
\end{theorem}

We can now prove Theorem \ref{thm:fpt_k}.

\begin{proof}[Proof of Theorem \ref{thm:fpt_k}]
    Given input $\left((G=(V,E),F),v_s,v_t,k,f\right)$ to \textsc{CCTO}, our algorithm proceeds as follows.
    \begin{enumerate}
        \item Compute an $(n-2,k)$-perfect hash family $\mathcal{F}$, using the algorithm of Theorem \ref{thm:hash-family}.
        \item For each function $g \in \mathcal{F}$, let $c_g$ be the corresponding colouring of $V \setminus \{v_1,v_2\}$.  Run the algorithm of Theorem \ref{thm:colour-solve} on the instance $\left((G,F),v_s,v_t,k,c_g,f\right)$ of \textsc{CCCTO}; if this algorithm returns YES, we return YES and halt.
        \item Return NO.
    \end{enumerate}
    It is immediate that, if the algorithm returns YES, we have a yes-instance: \textsc{CCCTO} simply introduces additional constraints compared with \textsc{CCTO}, so if we have a yes-instance to the former with any colouring this certainly corresponds to a yes-instance for the latter.  Conversely, suppose that $\left((G=(V,E),F),v_s,v_t,k,f\right)$ is a yes-instance of \textsc{CCTO}, and let $W$ be a walk witnessing this fact.  Let $v_1,\dots,v_{k-2}$ be the first $k-2$ vertices other than $v_s$ and $v_t$ on this walk (and note that $k-2$ such vertices must exist, since otherwise it would not visit the requisite number of vertices).  By definition of an $(n,k)$-perfect hash family, there must be some $g \in \mathcal{F}$ such that $c_g$ assigns a different colour to each of $v_1,\dots,v_{k-2}$.  It follows that $\left((G=(V,E),F),v_s,v_t,k,c_g,f\right)$ is a yes-instance for \textsc{CCCTO}, so our algorithm will return YES in the iteration corresponding to $c_g$.

    We now consider the running time.  Step 1 takes time $e^k k^{O(\log k)}n\log n$ by Theorem \ref{thm:hash-family}.  We repeat Step 2 a total of $e^k k^{O(\log k)}\log n$ times, and each repetition takes time $O(k!n^4T^4)$ by Theorem \ref{thm:colour-solve}.  Step 3 takes constant time.  The overall running time is therefore $k^{O(k)}T^4n^4\log n$.
\end{proof}

Note that if we assume that all costs are integral and positive, the number of objects we visit cannot exceed the budget by more than one. This gives us the following corollary.
\begin{corollary}\label{thm:fpt_fuel}
    If all costs assigned by $F$ are integral and positive, \textsc{CCTO} is in FPT with respect to $f$, the budget of fuel.
\end{corollary}

\section{Conclusion and Future Work}
We have introduced a generalisation of temporal graphs, named temporal cost graphs and used this to define the problem \textsc{Changing Cost Temporal Orienteering}. We have shown CCTO to be tractable when the underlying graph is a tree and the maximum traversal number of any edge under the function $F$ is restricted. We also show tractability when vertex-interval-traversal width, or number of vertices to be visited $k$ are restricted. These restrictions may not be reasonable in a real-world setting and so a natural step for future work could include an investigation of parameters which both give us tractability and are bounded in real-world applications of this model.  In particular, these could arise from the complex but periodic nature of distances between orbiting objects. 

As mentioned earlier, our algorithmic results can be applied to a version of the problem where the start and end are not specified. NP-completeness of this problems follows from the static graph setting, but we leave open whether hardness is retained for temporal cost graphs with maximum traversal number at most~4.

Another possible avenue for future work is to investigate approximation algorithms for CCTO. To the best of the authors' knowledge, the closest problem to CCTO for which there exists an approximation algorithm is Time-Dependent Orienteering, for which there exists a $(2+\varepsilon)$-approximation algorithm~\cite{FL02}. Their algorithm cannot be simply adapted to solve CCTO because their problem does not allow vertices to be revisited in a solution. Another possibility would be to start from the algorithm given by Chekuri et al.~\cite{chekuri2012improved} for the Orienteering problem where walks are allowed. However, this will likely require a significantly worse approximation factor or an exponential runtime caused by the addition of the fuel cost function to the problem.

\section*{Statements and Declarations}
The authors have no competing interests and no data was used in this work.
For the purpose of open access, the author(s) has applied a Creative Commons Attribution (CC BY) licence to any Author Accepted Manuscript version arising from this submission.

\subsection*{Author Contributions}
All authors (Timothée Corsini, Jessica Enright, Laura Larios-Jones, Kitty Meeks) contributed to the generation of ideas, the writing of the manuscript and the reviewing of the manuscript, with Corsini and Larios-Jones leading.

\subsection*{Funding}
        Timothée Corsini is supported by ANR TEMPOGRAL ANR-22-CE48-0001.
 Kitty Meeks is supported by EPSRC grants EP/T004878/1 and EP/V032305/1. Jessica Enright is supported by EPSRC grant EP/T004878/1. For the purpose of open access, the author(s) has applied a Creative Commons Attribution (CC BY) licence to any Author Accepted Manuscript version arising from this submission.

 \bibliography{references}
\end{document}